\documentclass[11pt]{article}
\usepackage{amsfonts}
\usepackage[
letterpaper=true,colorlinks=true,pdfpagemode=none,urlcolor=blue,linkcolor=blue,citecolor=blue,pdfstartview=FitH]{hyperref}

\usepackage{amsmath,amsfonts}
\usepackage{graphicx}
\usepackage{color}

\newif\ifblog
\newif\iftex
\blogfalse \textrue

\usepackage{ulem}
\def\em{\it}
\def\emph#1{\textit{#1}}

\newtheorem{theorem}{Theorem}
\newtheorem{lemma}[theorem]{Lemma}

\newtheorem{proposition}[theorem]{Proposition}
\newtheorem{example}{Example}
\newtheorem{remark}[theorem]{Remark}
\newenvironment{proof}{\noindent {\sc Proof:}}{$\Box$} 

\newcommand{\ep}{\varepsilon}

\newcommand{\p}{\partial}
\newcommand{\al}{\alpha}
\newcommand{\la}{\lambda}
\newcommand{\de}{\delta}
\newcommand{\be}{\begin{eqnarray}}
\newcommand{\ee}{\end{eqnarray}}
\newcommand{\bee}{\begin{eqnarray*}}
\newcommand{\eee}{\end{eqnarray*}}

\title{Indifference Pricing of American Option
  Underlying Illiquid Stock under Exponential Forward
  Performance}

\author{Xiaoshan Chen \thanks{Department of Mathematics,
    South China Normal University, {\tt
      sunscnu@126.com}. } \and
  Qingshuo Song \thanks{Department of Mathematics, City University of
    Hong Kong, {\tt songe.qingshuo@cityu.edu.hk}, The research of
    this author is
    supported in part by the Research Grants Council of Hong Kong
    No. CityU 104007. }
  \and Fahuai Yi   \thanks{Department of Mathematics,
    South China Normal University, {\tt
        fhyi@scnu.edu.cn}. } \and
  \and George Yin  \thanks{Department of
Mathematics, Wayne State University, Detroit, Michigan 48202, {\tt
gyin@math.wayne.edu}. Research of this author was supported in
part by the National Science Foundation under DMS-0907753, and in
part by the Air Force Office of Scientific Research under
FA9550-10-1-0210. }
}

\begin{document}
\date{}
\maketitle
\begin{abstract}
  This work focuses on  the
  indifference pricing of American call option
  underlying a non-traded stock, which
  may be partially hedgeable by another
  traded stock. Under the
  exponential forward measure,
  the indifference price is formulated
  as a stochastic singular
  control problem.
  The value function is characterized as the
  unique solution of a partial differential equation in
  a Sobolev space.
  Together with some regularities and estimates
  of the value function,
  the existence of the optimal strategy
  is also obtained.
  The applications of the characterization result
  includes a derivation of a dual representation
  and the indifference pricing on
  employee stock option.
  As a byproduct, a generalized
  It\^o's formula is obtained
  for functions in a Sobolev space.

\vskip 0.2 true in \noindent {\bf Keyword:} Stochastic control,
generalized verification theorem, portfolio optimization,
generalized It\^o's formula, indifference pricing, exponential
forward performance.

\vskip 0.2 true in \noindent {\bf AMS Subject Classification
Number.} 91G80, 93E20, 60J60.

\end{abstract}

\newpage

\section{Introduction}

Recent developments of indifference pricing, particularly the
valuation of American options under forward performance measures,
have attracted much attention; see \cite{LSZ11} and the references
therein. One of many applications of the indifference pricing is the
valuation of Employee Stock Option (ESO); see \cite{LS09, LS09b}. An
ESO is an American type of call option on the common stock of a
company, which is not available for trade in the market. However,
there may be some other partially correlated stocks (ex. Index)
available in the market, which may be used for the purpose of
partial hedge of ESO. Such a problem can be formulated as a
stochastic singular control problem \cite{MZ09} with optimal
stopping. Related study also appears in empirical research
\cite{BBL05, HL96} for the study of early exercise behavior.

In this paper, we
study the indifference pricing of the American call options
underlying a non-traded common stock
by assuming that there exists another perfectly liquid stock
available in the market, which may be used for the purpose of
partial hedge. The
indifference pricing involved is associated to a portfolio
optimization problem, which is in fact a stochastic control problem
with optimal stopping involved. Using the classical verification
theorem \cite{YZ99}, the value function may be characterized as the
unique solution of
an associated variational inequality (VI) under rather strong
assumptions. These assumptions are: (1) The VI has the unique
classical solution, i.e., unique solvability in $C^{2,2,1}$; (2) The
control problem has an underlying process and associated control
process that attains the  optimal value, i.e., existence of the
optimal pair $(\pi^*, \tau^*)$. We refer to \cite{FS06, YZ99}  for
the verification theorem in general control theory, and \cite{LSZ11}
for the case of indifference pricing of American options.

In the general setting of control problems,
the aforementioned conditions for the verification theorem are
difficult to
verify due to
the full non-linearity. As an alternative, one can use the viscosity
solution approach \cite{CIL92} to identify the value function. In
particular, \cite{OZ03} characterizes the value function $V$ of an
investment problem in the framework of indifference pricing. A
drawback of this approach is that due to the lack of information on
the regularity for the viscosity solution, one cannot obtain further
knowledge of optimal controls. Even the existence of the optimal
control is problematic.

In this work, we will adopt an intermediate methodology between the
classical verification theorem and the viscosity solution approach.
First, we show the verification theorem holds under weaker
assumptions than the classical verification theorem needs, but
stronger than the viscosity solution needs. For the purpose of
complete characterization, we obtain some regularity estimates of
the  partial differential equation (PDE), which verifies all the
necessary assumptions of the existence of the PDE solution and
associated optimal pair. Thanks to the above results on control
problem, one can further develop more useful results: They are the
regularity of optimal exercise boundary, the dual representation,
and monotonicity of the value function on the system parameters. To
be more precise, the idea is illustrated below.

In the standard arguments of the verification theorem, one uses
It\^o formula
and classical solutions of a variation inequality (VI).
However, it is too restrictive, since the VI in our case
does not admit a classical solution. Therefore, we replace the
assumption on the existence of classical solution by solutions in
Sobolev space $W_{p,loc}^{2,2,1}$  for large enough $p$. Then, the
question is
\begin{itemize}
\item  How large $p$ is enough to
enable us to use It\^o formula?
\end{itemize}
We will demonstrate that
$p=3$ is large enough. As a by product, we generalize It\^o's
formula given in \cite{Kry80} from one-side inequality to two-sided
equality; see the proof in the appendix, which may be potentially
useful in other related control problems.

Moreover, utilizing special structure with the exponential forward
performance, one can derive the verification for the price $P$
associated to a PDE.
Note that another assumption needed for the verification theorem,
namely the existence of optimal pair,
can be reduced to an estimation on the first derivative. As a
result, the verification theorem holds under weaker assumptions on
solvability (H1) and estimates (H2) given at the end of
Section~\ref{sec:vt}. To fully characterize the control problem, we
answer the questions:
\begin{itemize}
\item Does solvability (H1) and estimation (H2) hold for the
  VI associated to the price $P$?
\end{itemize}
To proceed, we start with a simple transformation on the backward
equation (in time), which leads to an equivalent counterpart of
forward equation on standard domain $\mathcal{Q} =
\mathbb{R}\times(0,T]$. Since the original domain is unbounded, we
shall use penalized method on the truncated version. As a result,
the truncated PDE after penalization is quasi-linear, and
Leray-Schauder fixed point theorem provides the solvability in
$W^{2,1}_p$, and the comparison principle is standard for the strong
solution, see more details in Section~\ref{sec:main}. By forcing the
limit of the parameter of the penalized function $\varepsilon\to 0$,
and the parameter of the truncated domain $N\to \infty$, we
use local estimate to obtain the $C^{\al,\al/2}$ estimate (De
Giorgi-Nash-Moser estimate) and  $W^{2,1}_p$ estimate on compact
domain by removing out singularity point $(\ln K, 0)$.
Consequently, we obtain  the existence of
$W^{2,1}_{p,loc}(\mathcal{Q}) \cap C(\overline{\mathcal{Q}})$
solution with some regularity estimates on the first order
derivatives, which eventually resolves (H1)-(H2). To this end, we
complete the characterization, which is summarized in
Theorem~\ref{thm:VP}.

One application of the main result is that, a dual representation
Proposition~\ref{prop:dual} is proved in the framework of stochastic
control theory. A simple derivation from this representation implies
several economically interesting facts: sensitivity properties
with respect to risk aversion $\gamma$ and some other parameters,
and sublinear property with respect to the payoff, etc. Another
application is the indifference pricing of ESOs, where vesting
periods and job termination risk is involved. Thanks to the
characterization
theorem, the price can be characterized as the classical solution of
its associated PDE.

The rest of this paper is outlined as follows. The precise
formulation of the problem is given in Section \ref{sec:form}.
Section \ref{sec:ver} proceeds with a generalized verification
theorem. Section~\ref{sec:main} presents the main results that gives
a complete characterization. Section \ref{sec:dual} is concerned
with the dual representation and other properties. Section
\ref{sec:appl} demonstrates application to ESO costs. Finally, an
appendix containing a few technical results is provided.

\section{Problem Formulation}\label{sec:form}

Let $(\Omega, \mathcal {F}, \mathbb{P})$ be a filtered probability
space with a filtration $(\mathcal {F}_t)_{t \geq 0}$ that satisfies
the usual conditions, and the processes $B$ and $\widetilde{B}$
be two independent standard  Brownian motions. Let
$\mathcal{T}_{t,T}$ be the set of all stopping times in the interval
$[t,T]$.

We consider a financial market with
a risk-free interest rate, which consists of a non-traded stock
issued by a firm with its price $Y^{t,y}_\nu$ at time $\nu \in
[t,T]$ given by
\begin{equation}
  \label{eq:Y}
  dY_\nu =b Y_\nu d \nu + c Y_\nu (\rho dB_\nu +
  \sqrt{1-\rho^2}d\widetilde{B}_\nu), \ Y_t = y,
\end{equation}
and one liquidly traded stock with its discounted price
$S^{t,s}_\nu$  at time $\nu \in [t,T]$
\begin{equation}
  \label{eq:stkS}
   dS_\nu=S_\nu \sigma (\lambda d \nu + d B_\nu), \
S_t = s.
\end{equation}

The payoff of American call option with strike $K>0$ and maturity
$T$ underlying non-traded asset $Y_\nu^{t,y}$ is $g(Y_\tau)$, where
$\tau$ is the
exercise time $\tau \in \mathcal{T}_{t,T}$ and $g$ is given by a
function given by
$$g(y) = (y- K)^+.$$

If an employee, with initial capital $w$, dynamically trades in the
stock $S$, then her  wealth process $W^{t,w,\pi}$ under
self-financing rule
satisfying
\begin{equation}
  \label{eq:W}
  dW^\pi_\nu = \pi_\nu\sigma(\lambda d\nu +dB_\nu), \ W(t) = w,
\end{equation}
where $\pi_\nu$ represents the  cash amount invested in the liquid
stock $S^{t,s}_\nu$. We assume the strategy $\pi$ belongs to
$\mathcal{Z}_{t}$, which is the set of all
progressively measurable processes $\pi:[t,T]\times \Omega \to
\mathbb{R}$ satisfying integrability condition
\begin{equation}
  \label{eq:pic}
  \mathbb{E}\Big[\int_t^T \pi_\nu^2 d\nu \Big] <\infty.
\end{equation}

We introduce an exponential forward performance measure given by
\begin{equation}
  \label{eq:fwdp}
  U_t(w) = - e^{-\gamma w + \frac 1 2 \lambda^2 t};
\end{equation}
see \cite[Theorem 3 and (18)]{LSZ11}. Under the above performance
measure, we compare the following two scenarios for an employee
holding the initial  capital $w$ and a unit of American call at the
initial time $t$:
\begin{enumerate}
\item Find a stopping time
  $\tau \in \mathcal{T}_{t,T}$ to exercise
  the call option, while dynamically trading
  the capital $w$ in liquid stock $S_\nu$, to maximize
  its performance, and find the corresponding value
  \begin{equation}
    \label{eq:V}
    V(w,y,t)= \sup_{\tau\in\mathcal {T}_{t,T}, \pi \in \mathcal
      {Z}_{t,\tau}} \mathbb{E} [U_\tau(W^{t,w,\pi}_\tau +
    g(Y^{t,y}_\tau))| \mathcal{F}_t].
  \end{equation}
\item
Receive a cash payment of amount $P(w,y,t)$
  by selling one unit of
  call at time $t$, the performance corresponding
  to the total  capital $w+P(w,y,t)$ is
  $U_t(w+ P(w,y,t))$.
\end{enumerate}

Now, we are ready to define the indifference price of the American
call option. The  indifference price of this call is the
 cash payment $P(w,y,t)$, which makes the above two
scenarios indifferent with respect to the given forward performance
of \eqref{eq:fwdp}. Therefore, the price $P(w,y,t)$ is the value
satisfying
\begin{equation}
  \label{eq:fwdpr}
  V(w,y,t) = U_t(w + P(w,y,t)).
\end{equation}

To determine the above indifference price of the American call, we
consider the following two problems:
\begin{enumerate}
\item[(Q1)] Find the indifference price
using the identity $P(w,y,t) = U_t^{-1}(V(w,y,t)) - w$,
  after solving the stochastic control optimization \eqref{eq:V},
  i.e.,
  \begin{enumerate}
  \item characterize the value function $V$ of \eqref{eq:V};
  \item characterize the pair of optimal strategy $\tau^*$ and $\pi^*$
    if it exists.
  \end{enumerate}
\item[(Q2)] Characterize directly the indifference price
   $P(w,y,t)$ as the unique solution of certain PDE.
\end{enumerate}
To proceed,  the following assumptions are imposed throughout this
paper:
\begin{enumerate}
\item [(A1)] Assume $\sigma$, $\gamma>0$, and $\rho\in [0,1)$.
\end{enumerate}
We exclude the case $\rho=1$, where the market is complete and the
problem of
indifference price can be reduced to
a Black-Scholes model.

\section{
Generalized Verification Theorem}\label{sec:ver} \label{sec:vt} We
present a version of the verification theorem for the value function
$V$ of \eqref{eq:V} and $P$ of \eqref{eq:fwdpr} in this subsection,
which generalize the  verification theorem from the classical PDE
solution to the strong solution in Sobolev function space with
certain regularity.

Define a three dimensional domain by $\mathcal{Q}^1 = \mathbb{R}
\times \mathbb{R}^+ \times [0,T)$, and denote by
$W^{2,2,1}_{p,loc}(\mathcal{Q}^1)$,
the collection of all functions $\mathcal{Q}^1 \ni (w,y,t) \mapsto
\varphi(w,y,t)$ having $(\partial_{ww} \varphi, \partial_{yy}
\varphi,
\partial_{t} \varphi)$ in distribution sense
and  integrable in any compact subset of  $\mathcal{Q}^1$. For a
scalar $\pi \in \mathbb{R}$, define a parameterized differential
operator $\mathcal{L}_1^\pi: W^{2,2,1}_{p,loc} \mapsto \mathbb{R}$
as
$$\mathcal{L}_1^\pi  \varphi (w, y, t) = \frac 1 2 c^2
y^2 \partial_{yy} \varphi + b y \partial_y \varphi + \frac{1}{2}
\sigma^2 \pi^2 \partial_{ww} \varphi + \pi (\rho\sigma c y
\partial_{wy} \varphi+\lambda \sigma \partial_w \varphi).$$
We also need to introduce   the {\it continuation region } by
\begin{equation}
  \label{eq:Q1}
  \mathcal{C}[v] = \{(w,y,t) \in \mathcal{Q}^1 : v(w,y,t) > U_t(w +
  g(y))\}.
\end{equation}

\begin{lemma} [Verification Theorem for
Function $V$]
  \label{lem:veri}
  Suppose there exists $v\in W_{p,loc}^{2,2,1}( \mathcal{Q}^1)$ for
  some $p\ge 3$, satisfying
  \begin{equation}
    \label{eq:pdeV}
    \min\big\{-\partial_t v + \inf_{\pi \in \mathbb{R}} \{ -
    \mathcal{L}_1^\pi v(w,y,t)\},
    v(w,y,t) - U_t(w + g(y))\big\} = 0,
  \end{equation}
  with the terminal condition
  \begin{equation}
    \label{eq:tcV}
    v(w,y,T) = U_T(w + g(y)).
  \end{equation}
  Then, $v(w,y,t) \ge V(w,y,t)$.
  If, in addition, there exists a
  pair $(W^*, \pi^*)$ satisfying \eqref{eq:W},
  \eqref{eq:pic}, and
  \begin{equation}
    \label{eq:pis}
    (\partial_t v + \mathcal{L}_1^{\pi^*_\nu} v) (W^*_\nu, Y_\nu, \nu)
    =  0, \ \forall\ t<\nu<\tau^*,
  \end{equation}
  where
  \begin{equation}
    \label{eq:taus}
    \tau^* := \inf \{ \nu>t: (W_\nu^*, Y_\nu, \nu) \notin
    \mathcal{C}[v]\}\wedge T.
  \end{equation}
  Then the variational inequality \eqref{eq:pdeV}-\eqref{eq:tcV}
  admits a unique solution in $W_{p,loc}^{2,2,1}(\mathcal{Q}^1)$,
  and $v(w,y,t) = V(w,y,t)$.
\end{lemma}

\begin{proof}
  Fix an arbitrary stopping time $\tau\in\mathcal {T}_{t,T}$ and an
  admissible control $\pi \in \mathcal {Z}_{t,\tau}$, and denote $(W,
  Y) := (W^{t,w,\pi}, Y^{t,y})$ and $\mathbb{E}_t[\ \cdot\ ] :=
  \mathbb{E}[\ \cdot\  | \mathcal{F}_t]$ for simplicity. By
  virtue of the
  generalized It\^o's formula (see Proposition~\ref{prop:Ito}),
  we have
  $$\mathbb{E}_t[v(W_\tau, Y_\tau, \tau)] = v(w,y,t) + \mathbb{E}_t
  \Big[ \int_t^\tau (\partial_t v + \mathcal{L}_1^\pi v) (W_\nu, Y_\nu,
  \nu) d\nu \Big].$$
  Applying \eqref{eq:pdeV},
  we have $v(w,y,t) \ge  U_t(w
  + g(y))$ for all $(w,y,t)$, hence
  $$\mathbb{E}_t[v(W_\tau, Y_\tau, \tau)] \ge
  \mathbb{E}_t[U_\tau(W_\tau + g(Y_\tau))].$$
  Note that \eqref{eq:pdeV} also implies
  $(\partial_t v + \mathcal{L}_1^\pi v) (w,y,t) \le 0$ for
  all $(w,y,t)$,
  which yields
  $$ \mathbb{E}_t
  \Big[ \int_t^\tau (\partial_t v + \mathcal{L}_1^\pi v) (W_\nu, Y_\nu,
  \nu) d\nu \Big] \le 0.$$
  Therefore, we conclude
  $$\mathbb{E}_t[U_\tau(W_\tau + g(Y_\tau))]
  \le  v(w,y,t),$$
  and arbitrariness of $\tau$ and $\pi$ further
  implies one-sided inequality
  $$v(w,y,t) \ge V(w,y,t).$$
  Furthermore, if
  \eqref{eq:pis} holds for some $(W^*, \pi^*)$, It\^o's
  formula together with the definition of $\mathcal{C}[v]$
  gives the inequality of the opposite direction:
  $$v(w,y,t) = \mathbb{E}_t[v(W_{\tau^*}, Y_{\tau^*}, \tau^*)] =
  \mathbb{E}_t[U_{\tau^*}(W_{\tau^*} + g(Y_{\tau^*}))] \le
  V(w,y,t).$$
  This implies $v = V$.
\end{proof}

Next, we discuss the price $P$ of \eqref{eq:fwdpr}. To proceed, we
introduce a domain $\mathcal{Q}^2 = \mathbb{R}^+ \times [0,T)$ and
its Sobolev
space $W^{2,1}_{p,loc}(\mathcal{Q}^2)$. Also define a non-linear
differential operator $\mathcal{L}_2: W^{2,1}_{p,loc}
(\mathcal{Q}^2) \to \mathbb{R}$ by
$$\mathcal{L}_2 \varphi(y,t) = \frac{1}{2} c^2
y^2 \partial_{yy} \varphi + (by-\rho c \lambda y) \partial_y \varphi
- \frac{1}{2} \gamma (1-\rho^2) c^2 y^2 (\partial_y \varphi)^2.$$

\begin{lemma}[Verification Theorem for Function $P$]
  \label{lem:vP}
  Suppose there exists a function
  $f\in W^{2,1}_{p,loc}(\mathcal{Q}^2)$ for some $p\ge 3$,
  satisfying
  \begin{equation}
    \label{eq:pdeP}
    \min\{ -\partial_t f(y,t) - \mathcal{L}_2 f (y,t),
    f(y,t) -g(y)\} = 0, \ \ (y,t) \in \mathcal{Q}^2,
  \end{equation}
  with the terminal condition
  \begin{equation}
    \label{eq:tcP}
    f(y,T) = g(y), \ \ y\in \mathbb{R}^+,
  \end{equation}
  and $|\partial_y f(y,t)|$ is uniformly  bounded.
  Then, PDE \eqref{eq:pdeV} together with the terminal condition
   \eqref{eq:tcV} is
  uniquely solvable  in $W^{2,2,1}_{p,loc}(\mathcal{Q}^1)$,
  and there exists   a pair
  $(W^*, \pi^*)$ satisfying \eqref{eq:W},
  \eqref{eq:pic}, and \eqref{eq:pis}. As a result,
  the price function $P(w,y,t)$
  of \eqref{eq:fwdpr} is independent to
  the initial wealth $w$, and $P(w,y,t) = f(y,t)$.
\end{lemma}

\begin{proof}
  Let $v(w,y,t) = U_t(w + f(y,t))$. Then, we
  have $v\in W_{p,loc}^{2,2,1}(\mathcal{Q}^1)$.
  One can directly check
  $\partial_{ww} v = \gamma^2 v<0$ in
  $\mathcal{Q}^1$, and hence, the function
  $\pi^*[v]: \mathcal{Q}^1 \mapsto \mathbb{R}$ of the form
  $$\pi^*[v](w,y,t) = - \frac{\rho c y \partial_{wy}v (w,y,t)
    +\lambda \partial_w v(w,y,t)}{\sigma \partial_{ww} v(w,y,t)}$$
  is well defined. Note that, $\pi^*[v]$ is independent to the
  variable $w$, i.e., one can rewrite
  $$\pi^*[v](y,t) = - \frac{\rho c }{\sigma} \cdot
  y \partial_y f
  (y,t)  +  \frac{\lambda}{\sigma \gamma}.$$
  Now, we can define a process $\pi^*$ by
  \begin{equation}
    \label{eq:pistar}
    \pi^*_\nu = \pi^*[v](Y_\nu, \nu) =
    - \frac{\rho c }{\sigma} \cdot
    Y_\nu \partial_y f (Y_\nu, \nu) +
    \frac{\lambda}{\sigma \gamma}, \ \ \nu\in (t,T).
  \end{equation}
  Note that, $Y$ is the unique strong solution of
  \eqref{eq:Y}. Therefore, $\pi_\nu^*$ is an admissible strategy,
  since it satisfies the integrability
  condition \eqref{eq:pic} due to
  boundedness of $\partial_y f$, i.e.,
  $$\mathbb{E}\Big[\int_t^T (\pi^*_\nu)^2 d\nu \Big] \le C + C
  \mathbb{E} \Big[  \int_t^T Y_\nu^2 d\nu \Big] <\infty.$$
  Since the optimal strategy $\pi^*_\nu$ is independent to
  $W^*$, the solution to \eqref{eq:W} can be simply given by its
  integral form of \eqref{eq:W} . Calculations by change of variables
  leads to
  \eqref{eq:pdeV}-\eqref{eq:pis}. Thus, Lemma~\ref{lem:veri} together
  with definition \eqref{eq:fwdpr} implies
  $V(w,y,t) = U_t(w+ f(y,t))$, and $f(y,t) = P(w,y,t)$.
\end{proof}

\begin{remark}[The boundary condition of $f(y,t)$
at $y=0$]{\rm
  The verification theorem can be thought of the probability
  counterpart of uniqueness result
  of  PDE solution \eqref{eq:pdeP}-\eqref{eq:tcP}.
  From the above verification theorem,
  the uniqueness holds without
  the specification of the boundary condition
  on $y=0$. Similar observations are addressed in
  \cite{OR73} (Fichera condition) for   linear
  equation, and in   \cite{BSY11} for
  the fully non-linear parabolic equation without
  obstacle.}
\end{remark}

To this end, we summarize the
implication of Lemma~\ref{lem:veri} and Lemma~\ref{lem:vP}. The
above verification theorems give conditional characterization of the
control problem since they give the answers to (Q1)-(Q2) based on
the assumptions  on the solvability of PDE. More precisely,
\begin{enumerate}
\item For (Q1-a), $V$ is the unique $W^{2,2,1}_{p,loc}$
  solution of PDE \eqref{eq:pdeV}-\eqref{eq:tcV};
\item For (Q1-b), the pair of optimal control exists,
  and they may have representation of
  \eqref{eq:pistar} and \eqref{eq:taus}, respectively;
\item For (Q2), $P$ is invariant in variable $w$,
  and the unique $W^{2,1}_{p,loc}$ solution
  of PDE  \eqref{eq:pdeP}-\eqref{eq:tcP} as a function
  of two variables of $(y,t)$,
\end{enumerate}
provided that the following hypothesis are valid,
\begin{enumerate}
\item [(H1)]    \eqref{eq:pdeP}-\eqref{eq:tcP} is
  solvable in
  $W_{p,loc}^{2,1}(\mathcal{Q}^2)$;
\item [(H2)] $|\partial_y f(y,t)|$ is uniformly bounded.
\end{enumerate}

\section{Main Result: Complete Characterization}\label{sec:main}
\label{sec:pde} To obtain the complete characterization, it is
crucial to study the solability and its related estimations of PDE
\eqref{eq:pdeP}-\eqref{eq:tcP}.

For the convenience in the analysis, we
analyze the backward equation (in time)
\eqref{eq:pdeP}-\eqref{eq:tcP},
by studying its associated forward equation of the following form.

Let $x=\ln y,\; \theta =T-t,\;u(x, \theta) = f(y,t)$ in
\eqref{eq:pdeP}-\eqref{eq:tcP}, then $u(x,\theta)$ satisfies
\begin{equation}
  \label{eq:pdeu}
  \left\{\begin{array}{ll}
    \min\{(\partial_\theta u - \mathcal{L} u) (x,\theta), \;
    u-(e^x-K)^+\}=0, \; & (x,\theta)\in \mathcal{Q},
    \vspace{2mm}\\
    u(x,0)=(e^x-K)^+, & x\in\mathbb{R},
  \end{array} \right.
\end{equation}
where $\mathcal{Q} = \mathbb{R}\times(0,T]$ and $\mathcal{L}$ is the
differential operator given by
$$\mathcal{L} \varphi (x, t) = \frac{1}{2} c^2 \partial_{xx} \varphi + \Big(b -
\rho c \lambda - \frac{1}{2}c^2\Big)\partial_x \varphi - \frac{1}{2}
\gamma(1-\rho^2)c^2(\partial_x \varphi )^2.$$

\subsection{Solvability of Problem \eqref{eq:pdeu}}
Since $(-\infty,+\infty)\times(0,T]$ is unbounded, we first confine
our attention to the truncated version of \eqref{eq:pdeu} in a
finite domain $\mathcal{Q}_N=(-N,N)\times(0,T]$. Let $u^N(x,\theta)$
be the solution (if it exists) to the following problem
\be\label{3.1} \left\{
  \begin{array}{ll}
    \min\big\{\p_\theta u^N - \mathcal{L}u^N, \;
    u^N-(e^x-K)^+\big\}=0,
    \;&(x,\theta)\in \mathcal{Q}_N,
    \vspace{2mm}\\
    \p_xu^N(-N,\theta)=0,\;\p_xu^N(N,\theta)=e^N,&\theta\in(0,T],
    \vspace{2mm}\\
    u^N(x,0)=(e^x-K)^+,&x\in(-N,N).
  \end{array}
\right. \ee

In order to prove the existence of solution to
problem (\ref{3.1}), we construct a penalty approximation of
problem (\ref{3.1}). Suppose $u^N_\ep(x,\theta)$ satisfies
 \be\label{3.2}
\left\{
  \begin{array}{ll}
    \p_\theta u^N_\ep- \mathcal{L} u^N_\ep +
    \beta_\ep(u^N_\ep-\pi_\ep(e^x-K))=0,\;
    &(x,\theta)\in \mathcal{Q}_N,
    \vspace{2mm}\\
    \p_xu^N_\ep(-N,\theta)=0,\;\p_xu^N_\ep(N,\theta)=e^N,&\theta\in(0,T],
    \vspace{2mm}\\
    u^N_\ep(x,0)=\pi_\ep(e^x-K),&x\in(-N,N),
  \end{array}
\right. \ee where $\beta_\ep(t)$ (see Fig. 1.) satisfies
\begin{eqnarray*}
  &&\beta _{\varepsilon }(t)\in C^{2}(-\infty ,+\infty ),\;\;\beta
  _{\varepsilon
  }(t)\leq 0, \\
  &&\beta _{\varepsilon }^{\prime }(t)\geq 0,\;\;\beta
  _{\varepsilon }^{\prime \prime }(t)\leq 0,\;\;\beta_\ep(0)=-C_0,
\end{eqnarray*}
where $C_0>0$ is to be determined.
Note that
\[
\lim\limits_{\varepsilon \rightarrow 0}\beta _{\varepsilon
}(t)=\left\{
  \begin{array}{ll}
    0, & t>0,\vspace{2mm} \\
    -\infty , & t<0,%
  \end{array}
\right.
\]%
and that $\pi _{\varepsilon }(t)$ (see Fig. 2.) satisfies
\[
\pi _{\varepsilon }(t)=\left\{
  \begin{array}{ll}
    t, & t\geq \varepsilon, \vspace{2mm} \\
    \nearrow, & |t|\leq \varepsilon, \vspace{2mm} \\
    0, & t\leq -\varepsilon,%
  \end{array}%
\right.
\]%
and $\pi _{\varepsilon }(t)\in C^{\infty },\;\;0\leq \pi
_{\varepsilon }^{\prime }(t)\leq 1,\;\;\pi _{\varepsilon }^{\prime
\prime }(t)\geq 0,\;\;\lim\limits_{\varepsilon \rightarrow 0^+}\pi
_{\varepsilon }(t)=t^{+}$. \vspace{1.3in}
\begin{center}
  \begin{picture}(0,-200)
    \put(-150,80){\vector(1,0){130}}
    \put(-90,10){\vector(0,1){95}}
    \qbezier(-102,15)(-85,80)(-80,80)
    \put(-10,80){$t$}
    \put(-80,85){$\varepsilon$}
    \put(-89,50){$-C_0$}
    \put(-90,57){\circle*{2}}
    \put(-80,80){\circle*{2}}
    \put(-130,-10){Fig. 1. \ \ $\beta_\varepsilon (t)$ }

    \put(0,40){\vector(1,0){155}}
    \put(70,10){\vector(0,1){95}}
    \qbezier(40,40)(68,37)(87,58)
    \put(70,40){\line(1,1){50}}
    \put(50,40){\circle*{2}}
    \put(90,40){\circle*{2}}
    \put(35,30){$-\varepsilon$}
    \put(85,30){$\varepsilon$}
    \put(160,40){$t$}
    \put(40,-5){Fig. 2. \ \ $\pi_\varepsilon (t)$ }
  \end{picture}
\end{center}

\begin{lemma}
  \label{lem:uen}
  There exists a unique solution $u^N_\ep(x,\theta)\in
  W^{2,1}_p(\mathcal{Q}_N)$ to
  problem (\ref{3.2}) for any $p\ge
  1$. Moreover, the following estimates hold,
  \be
  &&\pi_\ep(e^x-K)\leq u^N_\ep(x,\theta)\leq k\theta+x^2+e^{x+(b-\rho
    c\la)^+\theta}+1,\label{3.3}\\
  &&0\leq \p_xu^N_\ep(x,\theta)\leq e^{x+(b-\rho
    c\la)^+\theta},\label{3.4}\\
  &&\p_\theta u^N_\ep(x,\theta)\geq0,\label{eq:ptunep}
  \ee
  where $k\geq \max\left\{c^2+\frac{(b-\rho c\la-\frac{1}{2}c^2)^2}{2\gamma(1-\rho^2)c^2},
    \;c^2+\frac{\left[\gamma(1-\rho^2)c^2e^{(b-\rho c\la)^+T}-(b-\rho
        c\la-\frac{1}{2}c^2)\right]^2}{2\gamma(1-\rho^2)c^2}\right\}$.
\end{lemma}

\begin{proof} The
  Leray-Schauder fixed point theorem implies the existence of
  $W^{2,1}_p$ solution to
  problem (\ref{3.2}), and the comparison
  principle holds for the strong solution. The proof of the
  uniqueness is
  standard.

  First, we prove
  estimate (\ref{3.3}). Note that
  $u_1(x,\theta):=\pi_\ep(e^x-K)$
  satisfies $\partial_\theta u_1 - \mathcal{L} u_1+\beta_\ep(u_1-\pi_\ep(\cdot)) \le 0$ in
  $\mathcal{Q}_N$,  if we choose
  \begin{equation}
    \label{eq:C0}
    -\beta_\ep(0)=C_0=\rho c\la
    e^N+\frac{1}{2}\gamma(1-\rho^2)c^2e^{2N}.
  \end{equation}
  Moreover, when $N$ is large enough, we have
  $\p_xu_1(-N,\theta)=0,\;\p_xu_1(N,\theta)=e^N,\;u_1(x,0)=u^N_\ep(x,0).$
  Thus $u_1(x,\theta)=\pi_\ep(e^x-K)$ is a subsolution to
  problem (\ref{3.2}).

  On the other hand, one can check
  $u_2(x,\theta):=k\theta+x^2+e^{x+(b-\rho c\la)^+\theta}+1$ satisfies
  supersolution property
  $\partial_\theta u_2 - \mathcal{L} u_2+\beta_\ep(u_2-\pi_\ep(\cdot)) \ge 0$ in  $\mathcal{Q}_N$.
  Moreover, we have
  \bee
  &&\p_xu_2(-N,\theta)=-2N+e^{-N+(b-\rho c\la)^+\theta}\leq0,\\
  &&\p_xu_2(N,\theta)=2N+e^{N+(b-\rho c\la)^+\theta}\geq e^N,\\
  &&u_2(x,0)=x^2+e^x+1\geq \pi_\ep(e^x-K).
  \eee
  Applying the comparison principle,  we conclude \eqref{3.3}.

  Next, we
  verify inequality (\ref{3.4}). If we differentiate the
  equation in (\ref{3.2}) w.r.t. $x$, then
  $v(x,\theta)=\p_xu^N_\ep(x,\theta)$  satisfies
  \be\label{3.5}
  \left\{
    \begin{array}{ll}
      \p_\theta v-\frac{1}{2}c^2\p_{xx}v-(b-\rho c\la-\frac{1}{2}c^2)\p_xv\\
      +\gamma(1-\rho^2)c^2v\p_xv+\beta'_\ep(\cdot)(v-\pi'_\ep
      e^x)=0,\;&(x,\theta)\in \mathcal{Q}_N,
      \vspace{2mm}\\
      v(-N,\theta)=0,\;v(N,\theta)=e^N,&\theta\in(0,T],
      \vspace{2mm}\\
      v(x,0)=\pi'_\ep(\cdot)e^x,&x\in(-N,N).
    \end{array}
  \right.
  \ee
  Since $v_1 =  0$ and  $v_2(x,\theta)=e^{x+(b-\rho c\la)^+\theta}$ are
  subsolution and supersolution of \eqref{3.5}, the estimate
  \eqref{3.4} follows by comparison principle.

  Denote $u^\de(x,\theta):=u^N_\ep(x,\theta+\de)$ for $0<\de<T$,
  $u^\de(x,\theta)$ satisfies
  \bee
 \left\{
  \begin{array}{ll}
    \p_\theta u^\de- \mathcal{L} u^\de +
    \beta_\ep(u^\de-\pi_\ep(e^x-K))=0,\;
    &(x,\theta)\in (-N,N)\times(0,T-\de],
    \vspace{2mm}\\
    \p_xu^\de(-N,\theta)=0,\;\p_xu^\de(N,\theta)=e^N,&\theta\in(0,T-\de],
    \vspace{2mm}\\
    u^\de(x,0)=u^N_\ep(x,\de)\geq\pi_\ep(e^x-K)=u^N_\ep(x,0),&x\in(-N,N).
  \end{array}
 \right.
  \eee
  Combining with \eqref{3.2}, applying the comparison principle, we
  have
  \bee
  u^\de(x,\theta)\geq u^N_\ep(x,\theta),\;\;\;(x,\theta)\in (-N,N)\times(0,T-\de],
  \eee
  which implies \eqref{eq:ptunep}.
\end{proof}

\begin{lemma}
  \label{lem:un}
  Problem (\ref{3.1}) has a unique solution
  $u^N(x,\theta)\in W^{2,1}_{p,loc}(\mathcal{Q}_N)\cap
  C(\overline{\mathcal{Q}}_N)$ satisfying
  \be
  &&(e^x-K)^+\leq u^N(x,\theta)\leq k\theta+x^2+e^{x+(b-\rho
    c\la)^+\theta}+1,\label{3.6}\\
  &&0\leq \p_xu^N(x,\theta)\leq e^{x+(b-\rho
    c\la)^+\theta},\label{3.7}\\
  &&\p_\theta u^N(x,\theta)\geq0,\label{eq:ptun}
  \ee
  where $k$ is defined in Lemma~$\ref{lem:uen}$.
\end{lemma}

\begin{proof}
  Let $C_N$ be a generic constant independent of $\varepsilon$.

  Since
  $u^N_\ep(x,\theta)\geq\pi_\ep(e^x-K)$, then
  $|\beta_\ep(u^N_\ep-\pi_\ep(e^x-K ))|_{L^p(\mathcal{Q}_N)}
  \le C_N$.  One can treat the
  nonlinear term in the equation in (\ref{3.2}) as a linear term with
  the coefficient $\frac{1}{2}\gamma(1-\rho^2)c^2\p_xu^N_\ep$.
  Thanks to (\ref{3.4}),
  applying $C^{\al,\al/2}$ estimate (De Giorgi-Nash-Moser estimate,
  \cite{Lie96})   with $\alpha = 1/2$ , we have
  \bee
  |u^N_\ep|_{C^{1/2, 1/4}(\mathcal{Q}_N)}\leq C_N.
  \eee
  Letting
  $\ep\rightarrow0$, we have a continuous limit (of a subsequence if
  necessary) $u^N(x,\theta)$  up to the boundary, that is,
  \bee
  u^N_\ep(x,\theta)\rightarrow u^N(x,\theta) \;\;\;\;{\rm
    in}\;C(\overline{\mathcal{Q}}_N).
  \eee
  In view of (\ref{3.4}), applying $W^{2,1}_p$ estimate \cite{LSU67},
  we have
  \bee
  |u^N_\ep|_{W^{2,1}_p(\mathcal{Q}_N\setminus B_\de)}\leq C_N,
  \eee
  where $B_\de$ is a disk with center $(\ln K,0)$ and radius
  $\de>0$. Hence $u^N(x,\theta)\in W^{2,1}_{p,loc}(\mathcal{Q}_N)$  and
  \bee
  u^N_\ep(x,\theta) \rightharpoonup u^N(x,\theta) \;\;\;\; {\rm weakly \;
    in}\;W^{2,1}_{p,loc}(\mathcal{Q}_N),
  \eee
  which also implies $u^N(x,\theta)$ is
  a $W^{2,1}_{p,loc}$ solution
  of  (\ref{3.1}). Furthermore,
  (\ref{3.6})-(\ref{eq:ptun}) are consequences of (\ref{3.3})-
  (\ref{eq:ptunep}).

  In this below, we show the uniqueness.
  If not, we can find  two $W^{2,1}_{p}$ solutions  $u_1$ and $u_2$
  satisfying (\ref{3.6}) and (\ref{3.7}),
  and $\mathcal {N}=\{(x,\theta)\in
  \mathcal{Q}_N: u_1>u_2\}\neq\emptyset$. Observe that, $ \p_\theta u_1
  = \mathcal{L} u_1 $ and $ \p_\theta u_2 \ge  \mathcal{L} u_2$ in $\mathcal {N}$. Thus,
  $u_1-u_2$ satisfies
  \bee
  \left\{
    \begin{array}{ll}
      \p_\theta (u_1-u_2)-\frac{1}{2}c^2\p_{xx}(u_1-u_2)-(b-\rho
      c\la-\frac{1}{2}c^2)\p_x(u_1-u_2)\\
      +\frac{1}{2}\gamma(1-\rho^2)c^2(\p_xu_1+\p_xu_2)
      \p_x(u_1-u_2)\leq0,\;\;\;\;(x,\theta)\in\mathcal {N},
      \vspace{2mm}\\
      (u_1-u_2)(x,\theta)=0,\hspace{3.5cm}(x,\theta)\in\p_p\mathcal
      {N}\setminus\{x=\pm N\},
      \vspace{2mm}\\
      \p_x(u_1-u_2)(x,\theta)=0,\hspace{3.2cm}(x,\theta)\in\p_p\mathcal
      {N}\cap\{x=\pm N\}.
    \end{array}
  \right.
  \eee
  Owing to (\ref{3.7}), applying maximum principle (see \cite{Tso85}), we have
  $u_1-u_2\leq0$ on $\mathcal {N},$ which is a contradiction to  the
  definition of $\mathcal {N}$.
\end{proof}

\begin{lemma}
  \label{lem:u}
  There exists a unique solution
  $u(x,\theta)\in W^{2,1}_{p,loc}(\mathcal{Q}) \cap
  C(\overline{\mathcal{Q}})$ to
  problem \eqref{eq:pdeu} satisfying
  \be
  &&(e^x-K)^+\leq u(x,\theta)\leq k\theta+x^2+e^{x+(b-\rho
    c\la)^+\theta}+1,\label{3.11}\\
  &&0\leq \p_xu(x,\theta)\leq e^{x+(b-\rho
    c\la)^+\theta},\label{3.12}\\
  &&\p_\theta u(x,\theta)\geq0,\label{eq:ptu}
  \ee
  where $k$ is defined in Lemma~$\ref{lem:uen}$.
\end{lemma}

\begin{proof}
 First, we claim
 problem (\ref{3.1}) is
 equivalent to the following problem
 \be\label{3.9}
 \left\{
   \begin{array}{ll}
     \min\{\p_\theta u^N - \mathcal{L}u^N, \;
     u^N-(e^x-K)\}=0,
     \;&(x,\theta)\in \mathcal{Q}_N,
     \vspace{2mm}\\
     \p_xu^N(-N,\theta)=0,\;\p_xu^N(N,\theta)=e^N,&\theta\in(0,T],
     \vspace{2mm}\\
     u^N(x,0)=(e^x-K)^+,&x\in(-N,N).
   \end{array}
 \right.
 \ee
 In fact,  suppose $w(x,\theta)$ is a $W^{2,1}_{p,loc}$ solution to
 problem (\ref{3.9}),
 the maximum principle \cite{Tso85} implies  $w\geq 0$, combine with $w\ge
 e^x-K$ to get $w(x,\theta) \geq
 (e^x-K)^+$. Hence, $w$ is also a solution of
 \eqref{3.1}. Together with the uniqueness of \eqref{3.1}, the
 equivalence follows.

 Now we can rewrite
 problem (\ref{3.1}) as
 \be\label{3.10}
 \left\{
   \begin{array}{ll}
     \p_\theta u^N- \mathcal{L} u^N = f(x,\theta),\;&(x,\theta)\in \mathcal{Q}_N,
     \vspace{2mm}\\
     \p_xu^N(-N,\theta)=0,\;\p_xu^N(N,\theta)=e^N,&\theta\in(0,T],
     \vspace{2mm}\\
     u^N(x,0)=(e^x-K)^+,&x\in(-N,N),
   \end{array}
 \right.
 \ee
 where $f(x,\theta)=\chi_{\{u^N(x,\theta)=(e^x-K)^+\}}\big(-(b-\rho
 c\la)e^x+\frac{1}{2}\gamma(1-\rho^2)c^2e^{2x}\big)$.
 Thanks to (\ref{3.7}), if we apply $W^{2,1}_p$ interior estimate
 \cite{LSU67} to (\ref{3.10}) for arbitrary $M<N$,
 \bee
 |u^N|_{W^{2,1}_p(\mathcal{Q}_M\setminus B_\delta)}\leq C_M,
 \eee
 where $B_\delta$ is a disk with center $(\ln K,0)$ and
 radius $\delta$. We emphasize  $C_M$ only depends on $M$, but not on
 $N$.
 Fix $M>0$, and let $N\rightarrow+\infty$, $u^N$ leads to a limit
 $u^{(M)}$ (possibly a subsequence) in the fixed domain $\mathcal{Q}_M$ in
 the sense,
 \bee
 u^N\rightharpoonup u^{(M)}\;\;\;\;{\rm
 weakly\;in}\;W^{2,1}_{p,loc}(\mathcal{Q}_M)\;{\rm as}\;N\rightarrow\infty.
 \eee
 Moreover the sobolev embedding theorem implies
 \bee
 u^N\rightarrow u^{(M)}\;\;\;\;{\rm
 in}\;C(\mathcal{Q}_M),\;\p_xu^N\rightarrow \p_xu^{(M)}\;\;\;\;{\rm
 in}\;C(\mathcal{Q}_M).
 \eee

 It is clear that $u(x,\theta): = u^{(M)}(x,\theta),\;(x,\theta)\in\mathcal {Q}_M$ is well defined and $u$ is
 the solution to
 problem
 (\ref{eq:pdeu}). Moreover, $\p_xu\in C(\mathcal{Q})$ and we can
 deduce $u\in C(\overline{\mathcal{Q}})$ from the $C^\al$ estimate.
  (\ref{3.11})-(\ref{eq:ptu}) are consequences of
 (\ref{3.6})-
 (\ref{eq:ptun}). Lemma~\ref{lem:veri} and ~\ref{lem:vP} implies the
 uniqueness of solution to
 problem
 (\ref{eq:pdeu}).
\end{proof}

\subsection{The Obstacle of \eqref{eq:pdeu}}

Now we will characterize the optimal exercising time of $V$. Problem
\eqref{eq:pdeu} is an optimal stopping problem, which gives rise
 to a free boundary
 that can be expressed as a single-valued
 function of $\theta$. For later use, we define
 \bee
 &&\mathcal {S}:=
 \{(x,\theta)\in\mathcal {Q}: u(x,\theta)=(e^x-K)^+\} \;{\rm(Stopping\;
 region)},\\
 &&\mathcal {C}:=
 \{(x,\theta)\in\mathcal {Q}: u(x,\theta)>(e^x-K)^+\}\;{\rm(Continuation\;
 region)},
 \eee
where $u$ is the solution of \eqref{eq:pdeu}. Thanks to  the
continuity of $u$ from Lemma~\ref{lem:u}, $\mathcal{S}$ is closed
and $\mathcal{C}$ is open, respectively.
  \begin{lemma}
  \label{lem:sx}
  There exists $S(x): (-\infty,+\infty)\rightarrow
 [0,T]$ such that
 \be\label{4.1}
 \mathcal {S}=\{(x,\theta)\in\mathcal {Q}: 0<\theta\leq S(x)\}.
 \ee
 \end{lemma}
 \begin {proof}
  If $(x_0,\theta_0)\in\mathcal {S}$, i.e.,
 \bee
 u(x_0,\theta_0)=(e^{x_0}-K)^+,
 \eee
 according to \eqref{3.11} and \eqref{eq:ptu}, we have
 \bee
 u(x_0,\theta)=(e^{x_0}-K)^+, \; \theta\in(0,\theta_0].
 \eee
 Hence $\{x_0\}\times(0,\theta_0]\subseteq\mathcal {S}$, so we can
 define $S(x):  (-\infty,+\infty)\rightarrow
 [0,T]$ as
 \bee
 S(x)=
 \left\{
 \begin{array}{ll}
 0,\;\;\;{\rm if}\;u(x,\theta)>(e^x-K)^+\ {\rm
 for\;any}\;\theta\in(0,T],
 \vspace{2mm}\\
 \sup\{\theta\in(0,T]: u(x,\theta)=(e^x-K)^+\}.
 \end{array}
 \right.
 \eee
 By the definition of $S(x)$,
 (\ref{4.1}) holds.
 \end{proof}

 \begin{lemma}
 \label{lem:sxin}
 The free boundary $S(x)$ is
 strictly increasing w.r.t. $x$ on $\{x: 0<S(x)<T\}$.
 \end{lemma}

 \begin{proof}
 As in the proof of the equivalence between (\ref{3.1}) and
 (\ref{3.9}) in Lemma~\ref{lem:u}, we can show that problem \eqref{eq:pdeu} is
 equivalent to the following problem
 \be\label{3.14}
 \left\{
 \begin{array}{ll}
 \min\big\{\p_\theta u-\mathcal {L}u,\;u-(e^x-K)\big\}=0,\;&(x,\theta)\in\mathcal {Q},
 \vspace{2mm}\\
 u(x,0)=(e^x-K)^+,&x\in\mathbb{R}.
 \end{array}
 \right.
 \ee
  When
 $u(x,\theta)=(e^x-K)^+=(e^x-K)$, we have
 \bee
 \left\{
 \begin{array}{ll}
 e^x-K\geq 0,\label{4.2}\\
 \p_\theta u(x,\theta)-\mathcal {L}u(x,\theta)\geq0,\label{4.3}
  \end{array}
 \right.
 \eee
 which implies
 \be\label{4.4}
 e^x\geq \max\left\{K,\;
 \frac{2(b-\rho c\la)}{\gamma(1-\rho^2)c^2}\right\}.
 \ee

 For any $x_0\in\{x: 0<S(x)<T\}$, denote $S(x_0)=\theta_0\in(0,T)$,
 in view of (\ref{4.1}), we know
 \bee
 u(x_0,\theta)=e^{x_0}-K\geq0, \; \theta\in(0,\theta_0].
 \eee
 Denote $\mathcal {Q}_0=(-\infty,+\infty)\times(0,\theta_0]$ and define a new
 function $\overline{u}(x,\theta)$ on $\mathcal {Q}_0$ by
 \bee
 \overline{u}(x,\theta)=
 \left\{
 \begin{array}{ll}
 u(x,\theta),&(x,\theta)\in(-\infty,x_0]\times[0,\theta_0],
 \vspace{2mm}\\
 e^x-K,&(x,\theta)\in[x_0,+\infty)\times[0,\theta_0].
 \end{array}
 \right.
 \eee
 Since $\{x_0\}\times(0,\theta_0]\subseteq\mathcal {S}$, then
 $\overline{u}(x,\theta),\;\p_x\overline{u}(x,\theta)$ are continuous in
 $\mathcal {Q}_0$. Now we want to prove $\overline{u}(x,\theta)$ is the
 solution to
 problem \eqref{eq:pdeu} in the domain $\mathcal {Q}_0$.

 By the definition of $\overline{u}(x,\theta)$ we know
 \bee
 \overline{u}(x,0)=(e^{x}-K)^+, \; x\in\mathbb{R}.
 \eee
 According to \eqref{4.4}, we
 can check $\overline{u}(x,\theta)$
 satisfies
 \bee
 \min\big\{\p_\theta \overline{u}-\mathcal
 {L}\overline{u},\;\overline{u}-(e^x-K)^+\big\} =0, \; (x,t)\in\mathcal {Q}_0.
 \eee
 Hence, we know $\overline{u}(x,\theta)$
 is the solution  of \eqref{eq:pdeu} in
 domain $\mathcal {Q}_0$. By the
 uniqueness of $W^{2,1}_{p,loc}(\mathcal {Q})\cap C(\overline{\mathcal {Q}})$
 solution which satisfies (\ref{3.11})-(\ref{3.12}) to
 problem \eqref{eq:pdeu} we
 know
 \bee
 u(x,\theta)= \overline{u}(x,\theta),\;\;\;\;(x,\theta)\in\mathcal {Q}_0.
 \eee
 In particular,
 \bee
 u(x,\theta)= \overline{u}(x,\theta)=(e^{x}-K)^+,\;\;\;\;(x,\theta)\in[x_0,+\infty)\times(0,\theta_0].
 \eee
 By the definition of $S(x)$ we know
 \bee
 S(x)\geq \theta_0=S(x_0),\;\;\;\;x\geq x_0,
 \eee
 therefore the monotonicity of $S(x)$ is proved.

 Next we will prove the strict monotonicity of $S(x)$ on the set
 $\{x: 0<S(x)<T\}$. Suppose not (see Fig. 3). There exists $x_1<x_2$ such that
 \bee
 S(x)=\theta_0\in(0,T),\;\;\;\;x\in(x_1,x_2).
 \eee
 Thus
 \bee
 \left\{
 \begin{array}{ll}
 \p_\theta u(x,\theta)-\mathcal {L}u(x,\theta)=0,\;\;(x,\theta)\in(x_1,x_2)\times(\theta_0,T],
 \vspace{2mm}\\
 u(x,\theta_0)=(e^x-K),\;\;\;\;x\in(x_1,x_2).
 \end{array}
 \right.
 \eee
 Hence
 \be\label{4.5}
 \p_\theta u|_{\theta=\theta_0}&=&-\frac{1}{2}\gamma(1-\rho^2)c^2e^{2x}+(b-\rho
 c\la)e^x\nonumber\\
 &<&\Big[(b-\rho
 c\la)-\frac{1}{2}\gamma(1-\rho^2)c^2e^{x}\Big]e^{x_1}\nonumber\\
 &<&-\frac{1}{2}\gamma(1-\rho^2)c^2e^{2x_1}+(b-\rho
 c\la)e^{x_1}
 \leq0,
 \ee
 which contradicts (\ref{eq:ptu}), the two inequalities in (\ref{4.5}) is due to (\ref{4.4}).
 \end{proof}

 \begin{center}
 \begin{picture}(120,150)
 \thinlines
 \put(-120,70){\vector(1,0){150}}
 \put(-100,70){\vector(0,1){70}}
 \qbezier(15,140)(5,110)(-15,100)
 \qbezier(-80,70)(-60,78)(-45,100)
 \put(-45,100){\line(1,0){30}}
 \put(0,80){$\mathcal {S}$}
 \put(33,68){$x$}
 \put(-98,138){$\theta$}
 \put(-110,97){$\theta_0$}
 \multiput(-45,70)(0,5){6}{\line(0,1){3}}
 \multiput(-100,100)(5,0){12}{\line(1,0){3}}
 \multiput(-15,70)(0,5){6}{\line(0,1){3}}
 \put(-15,70){\circle*{2}}
 \put(-45,70){\circle*{2}}
 \put(-48,62){$x_1$}
 \put(-18,62){$x_2$}
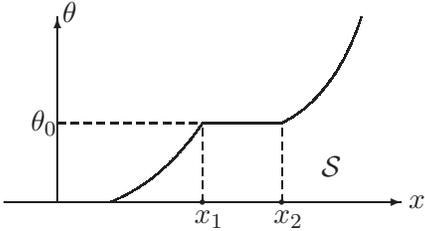
 \put(-128,45){Fig. 3. No strict monotonicity of $S(x)$}
 \put(140,65){\vector(1,0){110}}
 \put(140,65){\vector(0,1){85}}
 \qbezier(225,150)(216,125)(200,115)
 \qbezier(165,65)(190,73)(200,95)
 \put(253,63){$x$}
 \put(132,144){$\theta$}
 \put(130,91){$\theta_2$}
 \put(130,111){$\theta_1$}
 \multiput(200,65)(0,5){10}{\line(0,1){3}}
 \multiput(140,95)(5,0){12}{\line(1,0){3}}
 \multiput(140,115)(5,0){12}{\line(1,0){3}}
 \put(200,115){\circle*{2}}
 \put(198,57){$x_0$}
 \put(120,45){Fig. 4. Discontinuity of $S(x)$}
 \end{picture}
 \end{center}

\vspace{-2cm}
 \begin{lemma}
 \label{lem:sxco}
 The free boundary $S(x)$ is
 continuous on the set $\{x: 0<S(x)<T\}$.
 \end{lemma}

 \begin{proof} Suppose not (see Fig. 4). There
 would
 exist an $x_0$ such that
 \bee
 S(x_0):= \theta_1>\lim\limits_{x\rightarrow
 x_0^-}S(x):=\theta_2.
 \eee
 Then we
 would have
 \bee
 \p_\theta u(x_0,\theta)=\p_{x\theta}
 u(x_0,\theta)=0,\;\;\;\;\theta\in(\theta_2,\theta_1).
 \eee
 Since $\p_\theta u\geq0$ and in the domain $(x_0-\ep,x_0)\times(\theta_2,\theta_1)$
 \bee
 \p_\theta (\p_\theta u)-\frac{1}{2}c^2\p_{xx}(\p_\theta u)-\Big[b-\rho c\la-\frac{1}{2}c^2-\gamma(1-\rho^2)c^2\p_xu\Big]\p_x(\p_\theta
 u)=0.
 \eee
 Applying Hopf's Lemma \cite{Fri64} we obtain
 \bee
 \p_{x\theta}u(x_0,\theta)<0\;\;{\rm or}\;\;\p_\theta
 u\equiv0,\;\;(x,\theta)\in(x_0-\ep,x_0)\times(\theta_2,\theta_1),
 \eee
 which results in a contradiction.
 \end{proof}

 Since $S(x)$ is strictly increasing with respect to $x$ on the set
 $\{x: 0<S(x)<T\}$, then there exists an inverse function of $S(x)$,
 we denote it as $s(\theta)=S^{-1}(x),\;0<\theta<T$.
 Note that the strictly monotonicity of $S(x)$ is equivalent to the
 continuity of $s(\theta)$, and the continuity of $S(x)$ is equivalent
 to the strictly monotonicity of $s(\theta)$. Owing to Lemma~\ref{lem:sxin} and
 Lemma~\ref{lem:sxco}, we give the following theorem.

 \begin{lemma}
 \label{lem:fb}
 There exists an optimal exercising
 boundary $s(\theta): (0,T]\rightarrow \mathbb{R}$ such that
 \be\label{4.6}
 \mathcal {S}=\{(x,\theta)\in\mathcal {Q}: x\geq s(\theta)\}.
 \ee
 Moreover,  $s(\theta)$ is strictly increasing with respect to  $\theta$
 and
 \be
 &&s(\theta)\geq x_0,\label{4.7}\\
 &&s(0):= \lim\limits_{\theta\rightarrow0^+}s(\theta)
 =x_0,\label{4.8}
 \ee
 where $x_0=\left\{
 \begin{array}{ll}
 \ln K,&b-\rho c\la\leq0,
 \vspace{1mm}\\
 \max\{\ln K,\;\ln\frac{2(b-\rho c\la)}{\gamma(1-\rho^2)c^2}\},&b-\rho c\la>0.\end{array}
 \right.
 $
 In particular, $\p_\theta u$ is continuous across $s(\theta)$ and $s(\theta)\in C[0,T]\cap
 C^{\infty}(0,T]$.
 \end{lemma}

 \begin{proof} We can define
 \bee
 s(\theta)=
 \left\{
 \begin{array}{ll}
 S^{-1}(x),&0<\theta<T,
 \vspace{2mm}\\
 \inf\{x: S(x)=T\},&\theta=T.
 \end{array}
 \right.
 \eee
 Equation
 \eqref{4.6} is the consequence of the definition and monotonicity of
 $S(x)$.
 According to (\ref{4.4}) in the proof of Lemma~\ref{lem:sxin}, we know
 (\ref{4.7}) is true. Lemma~\ref{lem:sxin} implies $s(\theta)$ is monotonic increasing with
 respect to $\theta$, then we can define $s(0):=
 \lim\limits_{\theta\rightarrow0^+}s(\theta)$, and the proof of
 (\ref{4.8}) is similar to the proof of strictly monotonicity of
 $S(x)$ in Lemma~\ref{lem:sxin}.

 Moreover, since $\p_\theta u\geq0$ and
 $(e^x-K)^+$ is a lower obstacle, from \cite{Fried75} we know $\p_\theta u$ is
 continuous across $s(\theta)$ and $s(\theta)\in
 C^{\infty}(0,T]$.
 \end{proof}

\subsection{Main Result: Characterization}

Now, we are ready to present the complete characterization of the
value function $V$ and the indifference price $P$.
\begin{theorem}
  \label{thm:VP}
  \begin{enumerate}
  \item   The indifference price  of \eqref{eq:fwdpr} is independent of the
    initial capital $w$, and $P(w,y,t) := P(y,t)$ is the unique
    $W^{2,1}_{p,loc}(\mathcal{Q}^2)\cap C(\overline{\mathcal{Q}^2})$
    solution of     the variational inequality
    \eqref{eq:pdeP}-\eqref{eq:tcP}
    satisfying $|\partial_y P(y,t)|<C$ for some constant $C$.
  \item The value function of
    \eqref{eq:V}  has the form of
    $V (w,y,t) = U_t(w + P(y,t))$, and
    is the unique solution of
    \eqref{eq:pdeV}-\eqref{eq:tcV} in
    $W_{p,loc}^{2,2,1}( \mathcal{Q}^1)
    \cap C(\overline{\mathcal{Q}^1})$
    satisfying
    \begin{equation}
      \label{eq:cV}
      |\partial_y V(w,y,t)| < C |V(w,y,t)|.
    \end{equation}
  \item There exists $C^\infty$ function
    $y^*: [0,T) \to \mathbb{R}$, such that
    the strategy defined by
    $$\left\{
      \begin{array}{ll}
        \pi^*_\nu =
        - \frac{\rho c }{\sigma} \cdot Y_\nu \partial_y P (Y_\nu, \nu) +
        \frac{\lambda}{\sigma \gamma}, \ \ \nu\in (t,\tau^*),\\
        \tau^*=\inf\{\nu>t : Y_\nu\geq y^*(\nu)\} \wedge T,
      \end{array}\right.
    $$
    is optimal of the control problem \eqref{eq:V},
    in the sense that
    $$V(w,y,t) = \mathbb{E} [U_{\tau^*} (W^*_{\tau^*} +
    g(Y^{t,y}_{\tau^*}))| \mathcal{F}_t],$$
    where $W^*_{t_1} = w + \int_t^{t_1} \pi^*_\nu\sigma(\lambda d\nu
    +dB_{\nu}).$
  \end{enumerate}
\end{theorem}
\begin{proof}
  By the fact $y \p_y f (y,t) = \p_x u(x,\tau)$  and estimates of
  \eqref{3.12}, we conclude $|\p_y f(y,t)| \le C$ for some constant $C$.
  Applying  Lemma~\ref{lem:u}, Lemma~\ref{lem:vP}, together with
  Lemma~\ref{lem:veri} in order, we conclude (1)-(2) of
  Theorem~\ref{thm:VP}.  Thanks to the representations of
  \eqref{eq:taus} and \eqref{eq:pistar}, the optimal control
  can be written as
  $$\left\{
    \begin{array}{ll}
      \pi^*_\nu =
      - \frac{\rho c }{\sigma} \cdot Y_\nu \partial_y P (Y_\nu, \nu) +
      \frac{\lambda}{\sigma \gamma}, \ \ \nu\in (t,\tau^*),\\
      \tau^* = \inf \{ \nu>t: (W_\nu^*, Y_\nu, \nu) \notin
      \mathcal{C}[V]\} \wedge T.
    \end{array}\right.
  $$
  Note that the continuation region of \eqref{eq:Q1}
  satisfies
  $$
  \begin{array}{ll}
    \mathcal{C}[V] &=
    \{(w,y,t) \in \mathcal{Q}^1 : V(w,y,t) >
    U_t(w + g(y))\} \\
    & =  \{(w,y,t) \in \mathcal{Q}^1 :
    U_t(w + P(y,t)) >
    U_t(w + g(y))\} \\
    & =  \{(w, y, t) \in \mathcal{Q}^1 :
    P(y,t)) > g(y)\}.
  \end{array}
  $$
  Therefore, the obstacles of $P$ and $V$ are the same, and
  the optimal stopping time $\tau^*$ can be written invariant
  to $W^*$:
  $$ \tau^* = \inf \{ \nu>t: P(Y_\nu, \nu) \le
  g(Y_\nu)\} \wedge T.$$
  Take  $y^*(t)=e^{s(T-t)}$,
  where $s(\cdot)$ is the function in  Lemma~\ref{lem:fb}.
  Thanks to the
  the result of  Lemma~\ref{lem:fb},
  together with
  the transformation between $P(y,t)$ and $u(x,\tau)$,
  we conclude the representation of  $\tau^*$.
\end{proof}

\section{Dual Representation and
Ramifications}\label{sec:dual} In this part, we will give an
alternative proof of the result on dual representation of the
indifference price given in Proposition 7 of \cite{LSZ11}. Thanks to
the characterization Theorem~\ref{thm:VP}, the proof based on the
stochastic control approach is rather straightforward.

Now we recall the market with   stock prices given by \eqref{eq:Y}
and \eqref{eq:stkS}. Let $\mathcal{Z}$ be the collection of all the
processes of the form
$$Z_t^\varphi =
\exp\Big\{-\frac 1 2 \lambda^2 t - \lambda B_t\Big\}
\exp\Big\{-\frac 1 2 \int_0^t \varphi_\nu^2 d \nu - \int_0^t
\varphi_\nu d \widetilde B_\nu \Big\},
$$
where $\varphi$ is some $\mathcal{F}_\nu$  progressively measurable
process satisfying $\int_0^T \varphi_\nu^2 d \nu <\infty.$ Then, the
collection of equivalent local martingale measures $\Lambda$ can be
written by
$$\Lambda = \{\mathbb{Q}^\varphi:
d  \mathbb{Q}^\varphi = Z_T^\varphi d\mathbb{P}, \ Z^\varphi \in
\mathcal{Z} \}.$$ Among of them, we refer $\mathbb{Q}^0$ to the
minimal martingale measure (MMM). Under $\mathbb{Q}^\varphi$, two
processes
\begin{equation}
  \label{eq:bm2}
  B_t^\lambda := B_t + \lambda t
  \quad \hbox{ and } \quad
  \widetilde B_t^\varphi : = \widetilde B_t + \int_0^t
  \varphi_\nu d\nu
\end{equation}
are both standard Brownian motions. Let $\mathbb{Q}_t^\varphi$ and
$\mathbb{P}_t$ be the probability measures of $\mathbb{Q}^\varphi$
and $\mathbb{P}$ restricted to $\mathcal{F}_t$. Then, we have
$$d \mathbb{Q}_t^\varphi = Z_t^\varphi d \mathbb{P}_t.$$

To proceed, we introduce the following concept. The relative entropy
of a probability measure $\mathbb{P}_1$ with respect to
$\mathbb{P}_2$ is defined by
$$
H(\mathbb{P}_1|\mathbb{P}_2) = \left\{
  \begin{array}{ll}
    \displaystyle
    \mathbb{E}^{\mathbb{P}_1}\Big[ \log
    \frac{d \mathbb{P}_1}{d \mathbb{P}_2}\Big],
    & \mathbb{P}_1 <<\mathbb{P}_2,\\
    \infty, & \hbox{ otherwise.}
  \end{array}
\right.
$$

\begin{example}{\rm The relative entropy $\mathbb{Q}^\varphi$
  with respect to $\mathbb{P}$ is
  $$
  \begin{array}{ll}
    H(\mathbb{Q}^\varphi|\mathbb{P})
    & = \displaystyle
    \mathbb{E}^{\mathbb{Q}^\varphi} \Big[ \log
    \frac{d \mathbb{Q}^\varphi}
    {d \mathbb{P}} \Big]
    =     \mathbb{E}^{\mathbb{Q}^\varphi}
    [ \log Z_T^\varphi ] \\
    & =
    \displaystyle
    \mathbb{E}^{\mathbb{Q}^\varphi} \Big[
    - \frac 1 2 \lambda^2 T -
    \lambda B_T -
    \frac 1 2 \int_0^T \varphi_\nu^2 d \nu
    - \int_0^T \varphi_\nu d \widetilde B_\nu
    \Big] \\ & \displaystyle
    = \frac 12 \lambda^2 T +
    \mathbb{E}^{\mathbb{Q}^\varphi} \Big[
    \frac 1 2 \int_0^T \varphi_\nu^2 d \nu
    \Big].
  \end{array}
$$ }
\end{example}

\begin{example} \label{exm:re} {\rm Given a stopping time
  $\tau \in \mathcal{T}_{0,T}$,
  the relative entropy $\mathbb{Q}^\varphi_\tau$
  with respect to  $\mathbb{Q}^0_\tau$
  is
  $$
  \begin{array}{ll}
    H(\mathbb{Q}^\varphi_\tau | \mathbb{Q}^0_\tau)
    & = \displaystyle
    \mathbb{E}^{\mathbb{Q}^\varphi}
    \Big[ \log
    \frac{d \mathbb{Q}^\varphi_\tau}
    {d\mathbb{Q}^0_\tau} \Big]
    =     \mathbb{E}^{\mathbb{Q}^\varphi}
    \Big[ \log \frac{Z_\tau^\varphi}
    {Z^0_\tau} \Big] \\ \\
    & =
    \displaystyle
    \mathbb{E}^{\mathbb{Q}^\varphi} \Big[
    - \frac 1 2 \int_0^\tau \varphi_\nu^2 d \nu
    - \int_0^\tau \varphi_\nu d \widetilde B_\nu
    \Big] \\ \\ & \displaystyle
    =
    \mathbb{E}^{\mathbb{Q}^\varphi} \Big[
    \frac 1 2 \int_0^\tau \varphi_\nu^2 d \nu
    \Big].
  \end{array}
$$}
\end{example}

\begin{proposition}
  \label{prop:dual}
  The indifference price $P(w,y,t)$ of \eqref{eq:fwdpr}
  admits following dual representation for all
  $(w,y) \in \mathbb{R} \times \mathbb{R}^+$,
  \begin{equation}
    \label{eq:dP}
  P(w,y,0) = \displaystyle
  \hbox{\rm esssup}_{\tau \in \mathcal{T}_{0,T}}
  \hbox{\rm essinf}_{\mathbb{Q}^\varphi \in \Lambda}
  \left\{
    \mathbb{E}^{\mathbb{Q}^\varphi}
    [ g(Y_\tau^{y,0}) ] +
    \frac 1 \gamma
    H(\mathbb{Q}^\varphi_\tau | \mathbb{Q}^0_\tau)
  \right\}.
  \end{equation}
\end{proposition}
\begin{proof}
  Thanks to Example~\ref{exm:re}, it is enough
  to show that, for all $t\in [0,T]$
  $$
  P(w,y,t) = J(y,t) := \displaystyle
  \hbox{\rm esssup}_{\tau \in \mathcal{T}_{t,T}}
  \hbox{\rm essinf}_{\varphi \in \Phi}
  \left\{
    \mathbb{E}^{\mathbb{Q}^\varphi}_t
    \Big[ g(Y_\tau^{y,t})  +
    \frac 1 {2\gamma}
    \int_t^\tau \varphi_\nu^2 d\nu \Big]
  \right\},
  $$
  where $\Phi$ is the collection of
  all $\mathcal{F}_\nu$ progressively
  measurable process satisfying
  $\int_0^T \varphi_\nu^2 d \nu <\infty$,
  and
  $\mathbb{E}^{\mathbb{Q}^\varphi}_t[\ \cdot \ ] =
  \mathbb{E}^{\mathbb{Q}^\varphi}[\ \cdot \ | \mathcal{F}_t ]$.

  Recall that in \eqref{eq:bm2},
   both $B^\lambda$ and
  $\widetilde B^\varphi$ are Brownian motions under
  $\mathbb{Q}^\varphi$. We can
  rewrite the process $Y$ of
  \eqref{eq:Y} in terms of $B^\lambda$ and
  $\widetilde B^\varphi$,
  \begin{equation}
    \label{eq:Ylp}
    \begin{array}{ll}
      d Y_\nu & =
      (b - c\rho \lambda - c \sqrt{1- \rho^2} \varphi_\nu)
      Y_\nu d\nu + c\rho Y_\nu d B_\nu^\lambda +
      c \sqrt{1 - \rho^2} Y_\nu d \widetilde B_\nu^\varphi \\ \\
      & =
      (b - c\rho \lambda - c \sqrt{1- \rho^2} \varphi_\nu)
      Y_\nu d\nu + c Y_\nu d B_\nu^{\lambda, \varphi},
    \end{array}
  \end{equation}
  where $ B_\nu^{\lambda, \varphi}$ is another
  $\mathbb{Q}^\varphi$-Brownian motion.

  Since $\mathbb{Q}^\varphi \sim \mathbb{P}$ for arbitrary
  $\varphi \in \Phi$,
  \begin{center}
    $\int_0^T \varphi_\nu^2 d \nu <\infty$
    $\mathbb{P}$-almost surely
  \end{center}
  is equivalent to
  \begin{center}
    $\int_0^T \varphi_\nu^2 d \nu <\infty$
    $\mathbb{Q}^\varphi$-almost surely
  \end{center}
  and vice versa. In other words, there exists
  an 1-1 map $\Gamma: \Phi \mapsto \Phi$, such that
  the distribution of $\Gamma(\varphi)$ under $\mathbb{P}$
  is equal to the distribution of $\varphi$ under
  $\mathbb{Q}^\varphi$.
  Define
  \begin{equation}
    \label{eq:Yf}
    d Y^\varphi_\nu =
    (b - c\rho \lambda - c \sqrt{1- \rho^2} \varphi_\nu)
    Y^\varphi_\nu d\nu + c Y^\varphi_\nu d B_\nu.
  \end{equation}
  Comparing \eqref{eq:Yf} with \eqref{eq:Ylp},
  we conclude $Y$ under $\mathbb{Q}^\varphi$
  is equal to
  $Y^{\Gamma(\varphi)}$ under $\mathbb{P}$ in distribution.
  Therefore, we write $J(y,t)$ as a standard
  control problem of the following form,
  \begin{equation}
    \label{eq:con}
    \begin{array}{ll}
      J(y,t)
      &=
      \displaystyle
      \hbox{\rm esssup}_{\tau \in \mathcal{T}_{t,T}}
      \hbox{\rm essinf}_{\varphi \in \Phi}
      \left\{
        \mathbb{E}_t
        \Big[ g(Y_\tau^{\Gamma(\varphi), y,t})  +
        \frac 1 {2\gamma}
        \int_t^\tau \Gamma(\varphi)_\nu^2 d\nu \Big]
      \right\} \\
      & =
      \displaystyle
      \hbox{\rm esssup}_{\tau \in \mathcal{T}_{t,T}}
      \hbox{\rm essinf}_{\widetilde \varphi \in \Phi}
      \left\{
        \mathbb{E}_t
        \Big[ g(Y_\tau^{\widetilde \varphi, y,t})  +
        \frac 1 {2\gamma}
        \int_t^\tau \widetilde \varphi_\nu^2 d\nu \Big].
      \right\}
    \end{array}
  \end{equation}
  The second equality of \eqref{eq:con} follows from
  the fact $\Phi = \Gamma(\Phi)$.

  Applying exactly the same procedure of the
  verification theorem Lemma~\ref{lem:veri},
  we can conclude $J(y,t) = v(y,t)$,
  provided that there exists
  $v\in W_{3,loc}^{2,1}(\mathcal{Q}^2)$
  solves
  \begin{equation}
    \label{eq:pded}
    \left\{
      \begin{array}{ll}
        \min\Big\{-\partial_t v
        - \frac 1 2  c^2 y^2 \partial_{yy} v
        - (b - c\rho \lambda) y \partial_y v +
        \sup_{\varphi\in \mathbb{R}} \Big\{
        c \sqrt{1 - \rho^2} y \varphi \partial_y v
        - \frac 1 {2 \gamma} \varphi^2
        \Big\}, \\
        \hspace{3.0in} v(y,t) - g(y)
        \Big\} = 0, \ (y,t) \in \mathcal{Q}^2, \\
        v(y,T) = g(y), \ y\in \mathbb{R}^+.
      \end{array}\right.
  \end{equation}
  Thanks to  Theorem~\ref{thm:VP}(1),
  $P(w,y,t) = P(y,t)\in  W_{3,loc}^{2,1}(\mathcal{Q}^2)$
  solves \eqref{eq:pdeP}-\eqref{eq:tcP}.
  By utilizing
  the quadratic structure of
  $\sup_\varphi\{\ \cdot\ \}$ in \eqref{eq:pded}, one
  can check $P(y,t)$ also solves
  \eqref{eq:pded} . Therefore,
  $P(y,t) = v(y,t) = J(y,t)$.
\end{proof}

How does  the price $P(y,t)$ change, if we scale its payoff $g$ by
$n$ times, or if we change the risk aversion parameter $\gamma$? As
a result of the dual representation Proposition~\ref{prop:dual}, we
have the following properties on the price $P(y,t)$.

 \begin{proposition}
  \label{prop:est2}
  $P(y,t)$ decreases with respect to $\gamma$ and $\la$, and
  increases with respect to $b$,
  satisfying
  \be\label{eq:Pg}
  nP[g](y,t)\geq P[ng](y,t)\;\;(n\geq1),
  \ee
  where $P[\varphi]$ stands for the indifference
  price with the payoff function $\varphi$.
 \end{proposition}

 \begin{proof}
   Note that, Proposition~\ref{prop:dual} remains valid
   for the $P[ng]$, if we change $g$ into $ng$ in \eqref{eq:dP}.
   Since
   $H(\mathbb{Q}^\varphi_\tau | \mathbb{Q}^0_\tau)$ is non-negative,
   the monotonicity in $\gamma$ and the non-linearity of
   \eqref{eq:Pg}
   follow directly from Proposition~\ref{prop:dual}.

   Suppose $Y_i$ ($i = 1,2$) are two
   processes
   of \eqref{eq:Yf} related to $(b_i, \lambda_i)$.
   If $b_1\ge b_2$ and $\lambda_1\le \lambda_2$,
   then  $Y_1 \ge Y_2$ almost surely by
   comparison result of stochastic differential equation,
   and \eqref{eq:con} implies $P_1 \ge P_2$, where $P_i$ are
   the associated prices.
 \end{proof}

 An alternative proof of Proposition~\ref{prop:est2}
using a PDE approach is given in the Appendix, and it is
 interesting as its own right.
 In fact, Proposition~\ref{prop:est2}
 reveals natural economic facts.
 For instance, since $\gamma$ can be interpreted
 as the absolute risk aversion
 coefficient, proposition~\ref{prop:est2} claims that the employee's
 risk preference directly affect his exercise behavior.
 It implies a more prudent agent (with a larger coefficient of risk aversion
 $\gamma$) would be likely to exercise the option earlier to realize a
 cash benefit, and then invest it in other asset to
 earn the time value of the money obtained from the exercise of the
 option.
 So when he exercises, he
 gets less value than the risky agent,
 hence the value function decreases w.r.t. $\gamma$.
 In addition,
 \eqref{eq:Pg} implies that if the agent owns
 $n\ (n>1)$ pieces of options,
 compare with owning one piece of
 the option, the
 agent care less about the value of every piece
 of the option. So the average value of the
 indifference price about these pieces of
 options is less than the
 indifference price about one piece of this
 option.

As a straightforward consequence of \eqref{3.11} and \eqref{3.12},
we present the following estimates on $P(y,t)$ and $V(w,y,t)$.
\begin{proposition}
  \label{prop:est1}
  $P(y,t)$ and $V(w,y,t)$ satisfy
    \be
  &&(y-K)^+\leq P(y,t)\leq k(T-t)+(\ln y)^2+ye^{(b-\rho
    c\la)^+(T-t)}+1,\label{eq:pP}\\
  &&0\leq \p_yP(y,t)\leq e^{(b-\rho
    c\la)^+(T-t)},\label{eq:ppP}\\
  &&-e^{-\gamma(w+(y-K)^+)+\frac{1}{2}\la^2t}\leq V\leq -e^{-\gamma(w+k(T-t)+(\ln y)^2+ye^{(b-\rho
    c\la)^+(T-t)}+1)+\frac{1}{2}\la^2t},\label{eq:pV}\\
  &&0\leq \p_yV(w,y,t)\leq -\gamma e^{(b-\rho
    c\la)^+(T-t)}V(w,y,t),\label{eq:ppV}
  \ee
  where $k$ is defined in Lemma~\ref{lem:uen}.
\end{proposition}

\section{Application to ESO Costs}\label{sec:appl}
In this part, we
briefly describe the application of the above result
to Employee stock options (ESO). An ESO is a call option on the
common stock of a company, issued as a form of non-cash
compensations. Compared to the American call, the main differences
of ESO are the vesting period and job termination risk.

Suppose the company's stock price follows $Y$ of \eqref{eq:Y}, which
is non-tradable in the market. Given a unit of ESO with maturity
$T$, strike $K$, and vesting period $t_v \in (0,T)$, its payoff is
equivalent to the conventional American call only if the exercise
time occurs between $t_v$ and $T$, otherwise zero. Therefore, we can
write its payoff
as
$$(Y_\tau - K)^+ I_{\{\tau\ge t_v\}}$$
for the exercise time $\tau \in \mathcal{T}_{0,T}$. In addition, if
the job termination is taken into account, we have the following
revised  payoff,
$$(Y_{\tau\wedge \tau^\alpha} - K)^+
I_{\{{\tau\wedge \tau^\alpha} \ge t_v\}},$$ where $\tau^\alpha$ is
the time of employee's job termination. We assume $\tau^\alpha$ is a
random variable
having exponential distribution with parameter $\alpha>0$, which is
independent of Brownian motions $B$, and  $\widetilde B$.

As suggested by FASB rules, the price of ESO is evaluated by
risk-neutral measure, under which stock price $Y$ is a martingale.
For simplicity, we assume $b = 0$ in \eqref{eq:Y}, and $\mathbb{P}$
is the risk-neutral measure specified above. Following the arguments
of \cite[Section 5]{LS09b}, the indifference price has a
representation given as follows.
\begin{enumerate}
\item When ESO survives throughout the vesting period,
that is, if $t\ge t_v$, then the ESO cost $C(\cdot)$ satisfies
$$
\begin{array}{ll}
C(y,t) & = \mathbb{E}^{\bar Q} [ (Y_{\tau^*\wedge \tau^\alpha} -
K)^+]
\\
& = \mathbb{E}^{\bar Q} [ e^{-\alpha (\tau^*-t)} (Y_{\tau^*} - K)^+
+ \int_t^{\tau^*} \alpha (Y_\nu - K)^+ e^{-\alpha (\nu-t)} d \nu],
\end{array}
$$
where $\tau^*=\inf\{\nu>t : Y_\nu\geq y^*(\nu)\}\wedge T$ is optimal
exercise time in Theorem~\ref{thm:VP}. This corresponds to the
following PDE characterization: $C(y,t)$ is the unique $C^{2,1}$
solution of
  \begin{equation}
  \label{eq:cyt}
  \left\{
  \begin{array}{ll}
  \partial_t C + \frac 1 2 c^2 y^2 \p_{yy} C - \alpha C
  + \alpha(y - K)^+ = 0, \quad (y,t) \in \mathcal{Q}^2 \cap
  \{y < y^*(t)\} \cap \{t_v\le t<T\},\\
  C(y,t) = (y - K)^+, \quad (y,t) \in \mathcal{Q}^2 \cap
  \{y\geq y^*(t)\} \cap \{ t_v\le t<T\},\\
  C(y,T)=(y-K)^+,\quad y\in \mathbb{R}^+,
  \end{array} \right.
  \end{equation}
  since the boundary curve $y^*(t)$ is smooth due to
  Theorem~\ref{thm:VP} and PDE is non-degenerate locally.
\item If $t < t_v$ and ESO is alive at the moment, then the
  ESO cost is
  $$C(y,t) = \mathbb{E}^{\bar Q} [ C(Y_{t_v}, t_v)
  I_{\{\tau^\alpha  >t_v\}}]
  =  \mathbb{E}^{\bar Q} [e^{-\alpha (t_v-t)} C(Y_{t_v},t_v)].$$
  Thus, $C(y,t)$ follows,
  $$\partial_t C + \frac 1 2 c^2 y^2 \p_{yy} C - \alpha C = 0,\quad
  (y,t)\in\mathbb{R}^+\times[0,t_v),$$
  with terminal condition $C(y, t_v)$ given by
  the solution of \eqref{eq:cyt}.
\end{enumerate}

\noindent{\bf Acknowledgement} We thank Professor Nicolai V. Krylov
for the discussion on the generalized Ito's formula.

\appendix
\section{Appendix}\label{sec:app}

\subsection{A Generalized It\^o Formula}
We present a
generalized It\^o formula to the functions in Sobolev spaces
which is an extension of the result in \cite{Kry80}. We
thank Prof. Krylov
for confirming this result
by email communication.

\begin{proposition} \label{prop:Ito}
  Let $X$ be a diffusion on the filtered probability
  space
  $(\Omega,\mathcal{F}, \mathbb{P},
  \mathbb{F}= \{\mathcal{F}_t\})$, with
  generator $L$,  initial time $s$,
  and initial state $x\in \mathbb{R}^d$.
  Suppose $v\in W_{d+1,loc}^{2,1}(Q)$ for some open
  set $Q\subset \mathbb{R}^{d+1}$. Define
  $$\tau_Q^{s,x} = \inf\{r>s: X_r \notin Q\}.$$
  Then, for any $\mathbb{F}$-stopping time
  $\tau \le \tau_Q^{x,s}$, we  have
  \begin{equation}
    \label{eq:3}
    \mathbb{E}[v(X_\tau, \tau)] =
    v(x,s) + \mathbb{E} \Big[ \int_s^\tau L v(X_r,r) dr\Big].
  \end{equation}
\end{proposition}

\begin{proof}
  For any $\mathbb{F}$-stopping
  time $\tau \le \tau_Q^{s,x}$,
  by Theorem 2.10.2 of \cite{Kry80},
   we have
  \begin{equation}
    \label{eq:1}
    \mathbb{E}[v(X_\tau, \tau)] \le
    v(x,s) + \mathbb{E} \Big[
    \int_s^\tau L v(X_r,r) dr\Big].
  \end{equation}
  If we apply \eqref{eq:1} to a function $v_1 = -v$, then
  $$  \mathbb{E}[v_1(X_\tau,\tau)] \le
  v_1(x,s) + \mathbb{E} \Big[
  \int_s^\tau L v_1(X_r,r) dr\Big],
  $$
  which yields
  \begin{equation}
    \label{eq:2}
    \mathbb{E}[v(X_\tau, \tau)] \ge
    v(x,s) + \mathbb{E} \Big[
    \int_s^\tau L v(X_r,r) dr\Big].
  \end{equation}
  From \eqref{eq:1} and \eqref{eq:2},
  we conclude equality holds for
  \eqref{eq:1}.
\end{proof}

\subsection{Proof of Proposition~\ref{prop:est2}}
One can conclude the result of Proposition~\ref{prop:est2} by the
following two lemmas.
 \begin{lemma}
 \label{lem:ulaga}
 The solution $u(x,\theta)$ to the problem \eqref{eq:pdeu}
 decreases with respect to  $\gamma$ and $\la$, and increases
 with respect to $b$.
 \end{lemma}
 \begin{proof}
    We may prove all three monotonicities in the same
 way by comparison principle. Hence we only present the proof
 of the monotonicity
 of $u(x,\theta)$  with respect to  $\gamma$.
 Suppose $\gamma_1>\gamma_2$, and $u^N_{\ep(i)}(x,\theta)(i=1,2)$ is the solution to
 the following problem
 \bee
 \left\{
  \begin{array}{ll}
    \p_\theta u^N_{\ep(i)}-\frac{1}{2}c^2\p_{xx} u^N_{\ep(i)}-(b-\rho
    c\la-\frac{1}{2}c^2)\p_xu^N_{\ep(i)}\\
    +\frac{1}{2}\gamma_i(1-\rho^2)c^2(\p_xu^N_{\ep(i)})^2+
    \beta_{\ep}(u^N_{\ep(i)}-\pi_\ep(\cdot))=0,\;
    &(x,\theta)\in \mathcal{Q}_N,
    \vspace{2mm}\\
    \p_xu^N_{\ep(i)}(-N,\theta)=0,\;\p_xu^N_{\ep(i)}(N,\theta)=e^N,&\theta\in(0,T],
    \vspace{2mm}\\
    u^N_{\ep(i)}(x,0)=\pi_\ep(e^x-K),&x\in(-N,N),
  \end{array}
 \right.
 \eee
 Set $w(x,\theta):=u^N_{\ep(1)}(x,\theta)-u^N_{\ep(2)}(x,\theta)$. Then
 $w(x,\theta)$ satisfies
 \bee
 &&\p_\theta w-\frac{1}{2}c^2\p_{xx}w-(b-\rho
    c\la-\frac{1}{2}c^2)\p_xw
    +\frac{1}{2}\gamma_1(1-\rho^2)c^2(\p_xu^N_{\ep(1)}+\p_xu^N_{\ep(2)})\p_xw\\
    &&\qquad+
    \beta'_{\ep}(\cdot)w
 =\frac{1}{2}(\gamma_2-\gamma_1)(1-\rho^2)c^2(\p_xu_{\ep(2)}^N)^2\leq0.
 \eee
 Combining with the initial and boundary conditions, we have
 \bee
 u^N_{\ep(1)}(x,\theta)-u^N_{\ep(2)}(x,\theta)=w(x,\theta)\leq0.
 \eee
Letting $\ep\rightarrow0,\;N\rightarrow+\infty$, we know
 $u(x,\theta)$ is decreasing w.r.t. $\gamma$.
 \end{proof}

 \begin{lemma}
 \label{lem:ug}
 The solution to
 problem \eqref{eq:pdeu} satisfies
 \be\label{eq:ug}
 nu[\widetilde{g}(x)]\geq u[n\widetilde{g}(x)]\;\;(n\geq1),
 \ee
 where $\widetilde{g}(x)=(e^x-K)^+$, and $u[\widetilde{g}]$
 represents the solution to
 problem \eqref{eq:pdeu} with the obstacle
 and initial condition $\widetilde{g}$.
 \end{lemma}
 \begin{proof}
 Set $\widetilde{u}(x,\theta):=nu[\widetilde{g}(x)]$, then
 $\widetilde{u}(x,\theta)$ satisfies
 \bee
  \left\{
   \begin{array}{ll}
   \min\{\p_\theta \widetilde{u}- \mathcal{L}
     \widetilde{u}-\frac{1}{2}\gamma
     n(n-1)(1-\rho^2)c^2(\p_xu[\widetilde{g}])^2,\;\widetilde{u}-n\widetilde{g}\}=0,
     \vspace{2mm}\\
     \widetilde{u}(x,0)=n\widetilde{g}(x).
   \end{array}
 \right.
 \eee
 Note that $u[n\widetilde{g}(x)]$ satisfies
  \bee
  \left\{
   \begin{array}{ll}
   \min\{(\p_\theta u- \mathcal{L}
     u)(x,\theta),\;u-n\widetilde{g}\}=0,
     \vspace{2mm}\\
     u(x,0)=n\widetilde{g}(x),
   \end{array}
 \right.
 \eee
 We can confine the above problems in the bounded domain $\mathcal
 {Q}_N$. Suppose $\widetilde{u}^N[\widetilde{g}]$ and
 $u^N[n\widetilde{g}]$ are the solutions of the
 following problems
  \be\label{eq:nug}
  \left\{
   \begin{array}{ll}
   \min\big\{\p_\theta \widetilde{u}^N[\widetilde{g}]- \mathcal{L}
     \widetilde{u}^N[\widetilde{g}]-\frac{1}{2}\gamma
     n(n-1)(1-\rho^2)c^2(\p_xu^N[\widetilde{g}])^2,\\
     \hspace{7cm}\;\widetilde{u}^N[\widetilde{g}]-n\widetilde{g}(x)\big\}=0,\;(x,\theta)\in\mathcal
 {Q}_N,
     \vspace{2mm}\\
     \p_x \widetilde{u}^N[\widetilde{g}](-N,\theta)=0,\;\p_x
     \widetilde{u}^N[\widetilde{g}](N,\theta)=ne^N,\;\theta\in(0,T],
     \vspace{2mm}\\
     \widetilde{u}^N[\widetilde{g}](x,0)=n\widetilde{g}(x),\;x\in(-N,N).
   \end{array}
 \right.
 \ee

  \be\label{eq:ung}
  \left\{
   \begin{array}{ll}
   \min\big\{\p_\theta u^N[n\widetilde{g}]- \mathcal{L}
     u^N[n\widetilde{g}],\;u^N[n\widetilde{g}]-n\widetilde{g}(x)\big\}=0,\;(x,\theta)\in\mathcal
 {Q}_N
     \vspace{2mm}\\
     \p_x u^N[n\widetilde{g}](-N,\theta)=0,\;\p_x
     u^N[n\widetilde{g}](N,\theta)=ne^N,\;\theta \in(0,T],
     \vspace{2mm}\\
     u^N[n\widetilde{g}](x,0)=n\widetilde{g}(x),\;x\in(-N,N).
   \end{array}
 \right.
 \ee
 Comparing \eqref{eq:nug} with \eqref{eq:ung}, letting $N\rightarrow\infty$, we obtain
 \eqref{eq:ug}.
  \end{proof}

\bibliographystyle{plain}

\def\cprime{$'$}

\end{document}